\documentclass{amsart}

\usepackage{latexsym}
\usepackage{epsfig}
\usepackage{amsmath}
\usepackage{amsthm}
\usepackage{amsfonts}
\usepackage{exscale}
\usepackage{graphicx}

\newtheorem{prop}{Proposition}

\newtheorem{remark}{Remark}
\newtheorem{lemma}{Lemma}
\newtheorem{theorem}{Theorem}

\renewcommand{\Im}{\operatorname{Im}}

\newcommand{\sm}{\left(\begin{smallmatrix}}
\newcommand{\esm}{\end{smallmatrix}\right)}

\newcommand{\la}{\lambda}

\newcommand{\At}{\widetilde{A} }

\newcommand{\cN}{\mathcal{N}}
\newcommand{\cL}{\mathcal{L}}
\newcommand{\cO}{{\mathcal{O}}}

\newcommand{\eps}{\varepsilon}

\newcommand{\C}{\ensuremath{\mathbb{C}}}
\newcommand{\R}{\ensuremath{\mathbb{R}}}
\newcommand{\N}{\ensuremath{\mathbb{N}}}
\newcommand{\Z}{\ensuremath{\mathbb{Z}}}

\newcommand{\norm}[1]{\left\Vert#1\right\Vert}

\newcommand{\be}{\begin{equation*}}
\newcommand{\ee}{\end{equation*}}
\newcommand{\bea}{\begin{eqnarray*}}
\newcommand{\eea}{\end{eqnarray*}}
\newcommand{\ben}{\begin{eqnarray}}
\newcommand{\een}{\end{eqnarray}}
\newcommand{\beq}{\begin{equation}}
\newcommand{\eeq}{\end{equation}}
\newcommand{\enq}{\end{equation}}
\def\llangle{\left\langle}\def\rrangle{\right\rangle}

\title{On the spectrum of waveguides in planar photonic bandgap structures}

\author{B.M.Brown}
\address{Cardiff School of Computer Science,
Cardiff University, Cardiff, CF24 3AA, Wales, UK}
\email{Malcolm.Brown@cs.cardiff.ac.uk}
\author{V.Hoang}
\address{Institute for Analysis/Research Training Group 1294,
Karlsruhe Institute of Technology (KIT), Kaiserstrasse 89, Karlsruhe,
Germany}\email{ duy.hoang@kit.edu}
\author{M.Plum} \address{Institute for Analysis,
Karlsruhe Institute of Technology (KIT), Kaiserstrasse 89, Karlsruhe,
Germany}\email{michael.plum@kit.edu}
\author{I.Wood} \address{School of Mathematics, Statistics and Actuarial Sciences,
 University of Kent, Canterbury, CT2 7NF, UK}\email{ i.wood@kent.ac.uk}

\begin{document}

\begin{abstract}
We study a Helmholtz-type spectral problem related to the propagation
of electromagnetic waves in photonic crystal waveguides. The
waveguide is created by introducing a linear defect into a
two-dimensional periodic medium. The defect is infinitely extended
and aligned with one of the coordinate axes. The perturbation is
expected to introduce guided mode spectrum inside the band gaps of
the fully periodic, unperturbed spectral problem.
In the first part of the paper, we prove that, somewhat unexpectedly,
guided mode spectrum can be created by arbitrarily ``small'' perturbations.
Secondly we show that, after performing a Floquet decomposition in the
axial direction of the waveguide, for any fixed value of the quasi-momentum
$k_x$ the perturbation generates at most finitely many new eigenvalues inside the gap.
\end{abstract}

\maketitle

\section{Introduction}

The concept of a photonic crystal (also called a photonic band-gap
material) was suggested in 1987 (see e.g. \cite{Joann} for a
textbook introduction) and has received significant attention from both
the theoretical and experimental viewpoint.
In practice, photonic crystals are often manufactured using periodic
crystalline structures, a feature of which is their ability to allow,
or deny the propagation of electromagnetic radiation
which lies in a well defined range of the frequency spectrum.
 We call ranges where the electromagnetic radiation can
 propagate  spectral bands  and ranges where propagation is prevented  spectral gaps.
 These concepts will be made precise in Section \ref{formulation}.

One possible application of photonic crystals is their use for manufacturing highly
efficient optical waveguides which allow propagation of electromagnetic radiation only of very focussed frequencies.
 Briefly, these are created  by taking  some   photonic band-gap material,  called the bulk,
 and introducing a linear defect   which breaches the periodicity.
 This  may have the effect of  allowing propagation of electromagnetic radiation
 in a range of the frequency spectrum in which propagation is not possible in the bulk. Waves of these frequencies, sometimes called guided modes,  should then be highly focussed and almost entirely confined to the defect.

An  appropriate mathematical model for such materials  is given  by the spectral problem for
the Maxwell equations, or in the context of polarised waves in two
dimensions, by the Helmholtz equation.   In this paper, we shall study the
spectral problem for
$$
- \eps^{-1} \Delta,
$$
in $\R^2$, where $\eps$ is the dielectric function (or, equivalently,
the square of the refractive index of the material). We wish to model a waveguide in a  periodic  crystal and to do this  we take
$$\eps=\eps_0+\eps_1$$
where  $\eps_0$ is periodic and $\eps_1$ is a perturbation supported only in the waveguide,
which we choose as a strip in the $x$-direction.

In order to establish the existence of guided modes in a band gap of
the unperturbed medium (modelled by $- \eps_0^{-1}\Delta$), we must
prove that this gap  of the unperturbed problem now   contains
spectrum of $- \eps^{-1} \Delta$ induced by the perturbation
$\eps_1$. The mathematical investigation of these matters was begun
in \cite{AS, BHPW, KuchWav1, MiaoMa} and \cite{KuchWav2}, the latter
two articles working with the full 3D Maxwell equations. In
particular, \cite{KuchWav1, KuchWav2, MiaoMa} give sufficient
conditions for the existence of guided modes in the spectral gap of
the unperturbed problem, for example: sufficiently wide gaps always
contain spectrum of the perturbed problem. We note that existence of
gaps in the spectrum of some problems with periodic background media
was proved in \cite{FK96,CD99} and in \cite{Fil08} for the full
Maxwell case.

In this paper, we take another point of view and ask the following question: first  fix  a certain gap,
no matter how wide, is there a certain threshold strength of the
perturbation needed to produce spectrum in the gap? It turns out that
spectrum appears in the gap of the unperturbed problem  under  arbitrarily small perturbations.
We then further proceed to study the structure of the induced spectrum
for fixed quasimomentum $k_x$ in the direction
of the guided wave propagation. In particular, we show that for fixed $k_x$,
the eigenvalues introduced into the gap by the perturbation do not accumulate
at the ends of the gap interval.

It is interesting to compare our results to those available in the
literature. In \cite{FK97,FigKlein} the authors show generation of a
finite number of eigenvalues in the gaps by localized perturbations
of the periodic medium.   Another of their results in dimension
$d=3$ states that perturbations with $\eps_1 < 0$ (positive defects
in their terminology) need a certain threshold strength to create
eigenvalues in gaps. This is in sharp contrast to our result, which
uses specific features of the dimension two nature of the problem;
moreover, our defects are not localized. In \cite[Corollary 1]{AS},
a result concerning the existence of a threshold for a
two-dimensional waveguide analogous to that in \cite{FigKlein} is
claimed. This seems to be at odds with the result of Theorem
\ref{defectmodes} presented here.

Our paper is structured as follows: in section \ref{formulation}, we introduce the operators
to be studied and remind the reader briefly of the Floquet transform. Section \ref{bandedges}
contains some preparatory material on the band functions and Bloch functions,
which will play an important role later on. In section \ref{genofspec}, we use variational
arguments to prove the existence of guided mode spectrum. Finally, in section \ref{nonaccum},
we consider the question of non-accumulation of eigenvalues at the ends of the gap.



\section{Formulation of the problem}
\label{formulation}

We first consider the fully periodic background problem in $\R^2$. Let $\eps_0$ be a
positive bounded function on $\R^2$ which is bounded away from zero and which is periodic in both $x$ and $y$. For simplicity, we will assume that the basic cell of periodicity is $[0,1]^2$,
i.e.~that $\eps_0(x+1,y)=\eps_0(x,y)=\eps_0(x,y+1)$ for all $(x,y)\in\R^2$.
For any open set $O\subseteq\R^2$  we denote by $L^2_{\eps_0}(O)$ the weighted $L^2$-space with norm given by
$$\norm{u}_{\eps_0}^2=\int_{O}\eps_0(x)|u(x)|^2\ dx.$$

A 2D-periodic crystalline structure may be modeled by an
operator $L_0$ in $L^2_{\eps_0}(\R^2)$ given formally by
$$ L_0u=-\frac{1}{\eps_0(x,y)}\Delta u\quad \hbox{ with domain }
D(L_0)= H^2(\R^2).$$
This is a self-adjoint operator. A standard tool to analyse periodic
problems is the \emph{Floquet-Bloch} transform. We will state here
some of the results used in this paper and refer the reader to
\cite{KuchmentBook, RS} for proofs and more background on the
theory. The Floquet-Bloch transform $U_x$ associated with the
periodicity in the $x$-direction  is
\beq \label{Floquet} U_x :
L^2_{\eps_0}(\R^2) \to L^2_{\eps_0}(\Omega \times (-\pi,\pi)), \quad
U_x f(x, y, k) = (2\pi)^{-1/2} \sum_{m\in \Z} e^{i k m} f(x-m, y)
\enq where $\Omega:=(0,1)\times\R$ (more precisely, $U_x$ is defined
by the formula \eqref{Floquet} first for functions $f$ with
compact support, then extended to all of $L^2(\R^2)$). 

We now consider an operator family $L_0(k_x)$ on the strip $\Omega$ parametrised by $k_x\in [-\pi,\pi]$.
$L_0(k_x)$ is the self-adjoint operator in $L^2_{\eps_0}(\Omega)$ given by
\beq L_0(k_x)u=-\frac{1}{\eps_0(x,y)}\Delta u \eeq
defined on the space of all functions  $u\in H^2(\Omega)$ which satisfy  the
quasi-periodic boundary conditions
\beq \label{quasiperiodic}
u(1,y)=e^{i k_x} u(0,y)\quad\hbox{ and }\quad
\frac{\partial u}{\partial x}(1,y)=e^{i k_x} \frac{\partial u}{\partial x}(0,y)  .
\eeq
It follows from the general theory \cite{KuchmentBook, RS} that $L_0$ is the direct
integral of the operators $L_0(k_x)$:
\beq L_0 =
\int_{[-\pi,\pi]}^\bigoplus L_0(k_x) d k_x.
\eeq
As a consequence,
the spectrum of the problem in the plane is
\beq
\sigma(L_0)=\overline{\bigcup_{k_x\in [-\pi,\pi]}\sigma(L_0(k_x))}. \label{sigmat}
\eeq

Moreover, in view of periodicity in the $y$-direction, similar arguments apply
for each operator $L_0(k_x)$ and the spectrum of the operator $L_0(k_x)$ is
$$ \sigma(L_0(k_x)) =\overline{\bigcup_{k\in [-\pi,\pi]} \sigma(L_0(k_x, k))}$$
where $L_0(k_x, k)$ is the operator $-\frac{1}{\eps_0(x,y)}\Delta $ on the
unit cell $(0,1)^2$ with quasi-periodic boundary conditions in both the $x$- and $y$-directions:
\beq
\left\{\begin{array}{l}u(1,y)=e^{i k_x} u(0,y),\ \frac{\partial u}{\partial x}(1,y)=e^{i k_x} \frac{\partial u}{\partial x}(0,y),
\vspace{0.2cm}\\
u(x,1)=e^{i k} u(x,0),\ \frac{\partial u}{\partial y}(x,1)=e^{i k} \frac{\partial u}{\partial y}(x,0).
\end{array}\right.\label{both quasi-periodic}
\eeq
For notational simplicity we shall refer to the additional parameter as $k$ rather than $k_y$.

From \eqref{sigmat}, it is clear that any gap in the spectrum of  $L_0$ must arise from gaps in the spectrum of all  the operators $L_0(k_x)$. In this paper, we will fix $k_x$ and assume that there are gaps in the spectrum of  $L_0(k_x)$, an arbitrary one of which we denote by $(\mu_0,\mu_1)$. We will assume throughout that $\mu_0>0$.

We now turn to the problem which is our main interest in this paper. We perturb the original periodic problem in the plane to consider a wave-guide $W=\R\times (0,1)$ inside the 2D-periodic crystalline structure.
This new problem is modeled by an operator $L$ acting on $L^2_\eps(\R^2)$ given by
$$ Lu=-\frac{1}{\eps(x,y)}\Delta u$$ with $\eps(x,y)= \eps_0(x,y)+\eps_1(x,y)$.
The perturbation $\eps_1$ is supported in $W$, in the $x$-direction it is periodic with period 1 and it
is such that $\inf \eps_0+\eps_1>0$.

In view of the periodicity in the $x$-direction we can use the Floquet-Bloch transform \eqref{Floquet} to generate, as in the case of $L_0$, a selfadjoint operator family $L(k_x)$
 on the strip $\Omega$, acting in the space $L^2_\eps(\Omega)$,  parametrized by $k_x\in [-\pi,\pi] $ and given by
\be
L(k_x)u:= -\frac{1}{\eps_0+\eps_1}\Delta u
\ee
subject to the quasi-periodic boundary conditions \eqref{quasiperiodic}.
As before, the spectrum of the waveguide problem $L$ is
\beq\label{specL}
\sigma(L)=\overline{\bigcup_{k_x\in [-\pi,\pi]}\sigma(L(k_x))}.
\eeq
The chief goal of this paper is to compare the spectra of $L_0$ and $L$. In view of  \eqref{sigmat} and \eqref{specL},
this amounts to comparing the spectra of $L_0(k_x)$ and $L(k_x)$.
In \cite[Lemma 10]{AS} it is shown that, $\sigma(L(k_x))$ can differ from $\sigma(L_0(k_x))$
only through the introduction of extra eigenvalues.
Moreover, the essential spectra of the two operators coincide. In particular, the eigenvalues of $L(k_x)$ cannot accumulate at any point inside the spectral gaps of $L_0(k_x)$.

We shall therefore study the eigenvalues of the perturbed problem
\beq\label{eq:pert} -\Delta u = \lambda (\eps_0 + \eps_1) u\quad\hbox{ on } \Omega \eeq
where $\lambda\in (\mu_0,\mu_1)$, i.e.~$\la$ lies in a gap of the spectrum of the operator $L_0(k_x)$. It is understood that from now on all functions satisfy the
quasi-periodic boundary conditions \eqref{quasiperiodic}.
From the equation \eqref{eq:pert} we get
$$-\frac{1}{\eps_0} \Delta u -\lambda u = \lambda \frac{\eps_1}{\eps_0} u.$$
It then follows that $\lambda$ is an eigenvalue in the gap iff 
\beq u=\lambda \left(L_0(k_x) -\lambda\right)^{-1} \left( \frac{\eps_1}{\eps_0}u\right) \label{FP} \eeq
holds
for some non-zero $u$. Our strategy will be based on finding solutions of \eqref{FP} using information on the resolvent of the unperturbed operator $L_0(k_x)$.

As mentioned before, we shall show that small perturbations of
the operator $L_0(k_x)$ create extra spectrum in an arbitrary fixed gap.
Our second aim in this paper is to prove that
the additional eigenvalues do not accumulate anywhere on the closure
of the spectral gap $(\mu_0,\mu_1)$, in particular not at the
endpoints $\mu_0$ and $\mu_1$.

We consider the non-accumulation result as a first step towards a
better understanding of the structure of the guided mode spectrum.
One could surmise that the guided mode spectrum $\sigma(L(k_x))\setminus\sigma(L_0(k_x))$
can be written in terms of continuous band functions $\{\theta_j(k_x)\}$ depending on $k_x\in [-\pi,\pi]$.
However, since eigenvalues may be emitted and absorbed by the essential spectrum of the
unperturbed operator $L_0(k_x)$ as $k_x$ varies, the functions $\theta_j$ may possibly be defined
only on subintervals of $[-\pi, \pi]$. An open question is whether a finite total number
of band functions $\theta_j$ (possibly defined on subintervals of $[-\pi,\pi]$) is sufficient to describe
the guided mode spectrum.

\section{Eigenvalues at band edges and their eigenfunctions}
\label{bandedges}

Our analysis of the waveguide problem  will be via a study of the resolvent of the unperturbed operator $L_0(k_x)$. This will be performed using the Floquet-Bloch transform. In this section we introduce a representation of the resolvent in terms of so-called Bloch functions and show some results which will be crucial in our later analysis. We note that the perturbation plays no role in this section.

For fixed $k_x$ we consider the operator 
$L_0(k_x, k)$ on $(0,1)^2$ introduced in the previous section. The operator depends on $k$ via the quasi-periodic boundary conditions and thus has a $k$-dependent domain. We can transform
the eigenvalue problem for $L_0(k_x, k)$ into an eigenvalue problem for a $k$-dependent operator with periodic
boundary conditions in the $y$-direction:
Let $\nabla_k=\nabla+i\left(\begin{array}{c}0\\k\end{array}\right)$
and define $\Delta_k:=\nabla_k^2$ on $(0,1)^2$ subject to the boundary conditions
\beq
\left\{\begin{array}{l}
u(1,y)=e^{i k_x} u(0,y),\ \frac{\partial u}{\partial x}(1,y)=e^{i k_x} \frac{\partial u}{\partial x}(0,y),
\vspace{0.2cm}\\
u(x,1)= u(x,0),\ \frac{\partial u}{\partial y}(x,1)= \frac{\partial u}{\partial y}(x,0).
\end{array}\right.\label{mixed}
\eeq

Then the self-adjoint analytic operator family $-\frac{1}{\eps_0}\Delta_k$ in $L^2_ {\eps_0}((0,1)^2)$
is of type (A) (cf. \cite{Kato}), i.e. the domain of the operators does not vary
when $k$ varies.
Further, it is shown in  \cite[Theorem VII.3.9]{Kato} that there exist collections of functions $\{\lambda_s(k)\}_{s\in \N}$ and $\{\phi_s(k)\}_{s\in \N}$
which are real-analytic functions in the variable $k$ on $[-\pi,\pi]$. Moreover,
for each $s\in \N$, $\lambda_s$ and $\phi_s$
can be continued analytically to an open set
$$\{ z \in \C: \operatorname{Re}\ z \in (-\pi-\delta_s, \pi+\delta_s),\ |\Im\ z | < \eta_s \}$$
containing the interval $[-\pi,\pi]$. For every $k\in [-\pi,\pi]$, the function
$\phi_s(k)\in L^2_{\eps_0}((0,1)^2)$ is
 a normalized eigenfunction of $-\frac{1}{\eps_0}\Delta_k$ satisfying \eqref{mixed} and $\lambda_s(k)$
is the corresponding eigenvalue. We note that the eigenvalues are not necessarily ordered
by magnitude. Moreover, the analyticity results depend critically on
the fact that we only let the scalar parameter $k$ vary.
The normalized eigenfunctions $\{\psi_s(x,y,k)\}_{s\in \N}$ of $L_0(k_x, k)$ are then given by
$\psi_s(x,y,k)=e^{iky}\phi_s(x,y,k)$, with eigenvalues $\lambda_s(k)$. We call these the Bloch functions.

\begin{lemma}
  Let $\norm{u}_2=\left(\int_{(0,1)^2}|u|^2\right)^{1/2}$ denote the unweighted $L^2$-norm of a function over $(0,1)^2$. We have the following gradient estimates:
  \beq 
 \norm{ \nabla \phi_s}_2 \leq C(\sqrt{ \lambda_s(k)}+1) \quad\hbox{ and } \quad  \norm{ \nabla \psi_s}_2 \leq \sqrt{ \lambda_s(k)}. \label{psi_s_estimate}
 \eeq
\end{lemma}

\begin{proof}
 As $-\frac{1}{\eps_0}\Delta_k \phi_s =\lambda_s(k)\phi_s(k)$, testing with $\phi_s$, gives
 $$\int_{(0,1)^2} | \nabla_k\phi_s|^2=\int_{(0,1)^2}\left(-\frac{1}{\eps_0}\Delta_k\phi_s\right)\overline{\phi_s}\eps_0= \int_{(0,1)^2} \lambda_s(k)\phi_s\overline{\phi_s}\eps_0=\lambda_s(k).$$

 Hence,  \bea
 \norm{ \nabla \phi_s}_2=\norm{\nabla_k\phi_s-i\left(\begin{array}{c}0\\k\end{array}\right)\phi_s}_2 \leq \norm{\nabla_k\phi_s}_2 + |k| \norm{\phi_s}_2 \leq C(\sqrt{ \lambda_s(k)}+1).
 \eea
 Note that, as $k$ runs through a bounded set, the constant $C$ can be chosen independent of $k$, but that it depends on $\norm{1/\eps_0}_\infty$.
 The second statement follows from integration by parts:
 $$\int_{(0,1)^2} | \nabla\psi_s|^2=\int_{(0,1)^2}\left(-\frac{1}{\eps_0}\Delta\psi_s\right)\overline{\psi_s}\eps_0= \int_{(0,1)^2} \lambda_s(k)\psi_s\overline{\psi_s}\eps_0=\lambda_s(k).$$\end{proof}

We remind the reader that $\mu_1$ is the lowest point of a spectral band and lies at the top end of a gap.
The next result shows that only finitely many spectral bands can touch $\mu_1$ and that each
$\lambda_s(k)$ can touch $\mu_1$ only finitely many times.
\begin{prop}\label{finite}
There are only finitely many pairs $(s_i,k_i)$ such that
$\lambda_{s_i}(k_i)=\mu_1$.
\end{prop}

\begin{proof}
Let $k,k_0\in[-\pi,\pi]$. Arguing as in \cite[VII.3.6]{Kato}, we get that there exist
$s$-independent constants $C_1$, $C_2$ such that
$$|\lambda_s(k)-\lambda_s(k_0)|\leq (C_1+|\la_s(k_0)|)( e^{C_2|k-k_0|}-1).$$
Next, choose $|k-k_0|<\delta$, where $\delta$ is such that $e^{C_2\delta}-1<1/2$. Then
$$|\lambda_s(k)-\lambda_s(k_0)|\leq \frac12(C_1+|\la_s(k_0)|)$$
and so whenever $\la_s(k_0)\geq0$ we get
 $$\lambda_s(k)\geq \frac12\lambda_s(k_0)-C.$$
For fixed $k_0$, we
have $\lambda_s(k_0)\to\infty$ as $s \to\infty$. Hence $$\inf_{\{k:|k-k_0|<\delta\}}
\lambda_s(k)\geq \frac12\lambda_s(k_0)-C\to\infty\quad\hbox{ as } s\to \infty.$$
Since $\delta$ can be chosen independently of $s$ and $k_0$, we can cover $[-\pi,\pi]$ with finitely many intervals of length $\delta$ and so $
\inf_{k\in[-\pi,\pi]}
\lambda_s(k)\to\infty$ as $ s\to \infty.$
Therefore, there are only finitely many values of $s$ such that
$\lambda_{s}(k)=\mu_1$ for some $k$.

On the other hand, if for any fixed $s$ we have
$\lambda_s(k_j)=\mu_1$ for infinitely many $k_j$, then these must
accumulate in $[-\pi,\pi]$ and the analyticity of $\lambda_s$ would
imply $\lambda_s(k)\equiv \mu_1$. But applying the well-known
Thomas argument \cite{Thomas, KuchmentBook},
we see that none of the functions $\lambda_s$ can be constant.
\end{proof}

We now discuss the expansion of functions in $L^2_{\eps_0}((0,1)^2)$ in terms of the Bloch functions $\psi_s$.
Since for each $k$, the eigenfunctions $\psi_s(\cdot, k)$ form a complete orthonormal system
in $L^2_{\eps_0}((0,1)^2)$, we have for  $r\in L^2_{\epsilon_0}(0,1)^2$ and $k\in  (-\pi,\pi)$
$$
r = \sum_{s\in\N} nm, \llangle r, \psi_s(\cdot, k) \rrangle_{L^2_{\eps_0}((0,1)^2)} \psi(\cdot, k)
$$
and hence by using Parseval's identity and integrating over $k$, we get
\beq\label{complete}
\norm{r}_{L^2_{\eps_0}((0,1)^2)}^2= \frac{1}{2\pi} \sum_{s\in\N}\int_{-\pi}^\pi |\llangle r, \psi_s(\cdot,k)\rrangle_{L^2_{\eps_0}((0,1)^2)}|^2 \ dk.
\eeq

We now derive a representation for $(L_0(k_x) -\la)^{-1}r$ with $r\in L^2_{\eps_0}((0,1)^2)$
in terms of Bloch functions. \emph{Whenever we apply an operator with domain $L^2(\Omega)$ to
a function $r\in L^2_{\eps_0}((0,1)^2)$, we extend $r$ to all of $\Omega$ by zero}. We will use
the same letter for the extended function.

Let $U$ denote the Floquet transform in the $y$-variable defined analogously to \eqref{Floquet} and set
\beq\label{eq:Pskr}
P_s(k,r):=\llangle Ur(\cdot,k),\psi_s(\cdot,\overline{k})\rrangle_{L^2_{\eps_0}((0,1)^2)}\psi_s(\cdot,k).
\eeq
We note that since $r$ is supported in $[0,1]^2$, we  have
\beq\label{eq:Pskr1}
P_s(k,r)=\frac{1}{\sqrt{2\pi}}\llangle r,\psi_s(\cdot,\overline{k})\rrangle_{L^2_{\eps_0}((0,1)^2)}\psi_s(\cdot,k).
\eeq
In view of the analyticity of the Bloch functions $\psi_s$, it follows that also $P_s$ depends analytically
on $k$ in a small complex neighborhood of $[-\pi,\pi]$ (note that $k \mapsto \overline{\psi_s(\cdot, \overline{k})}$
is analytic in $k$, which can be seen by expansion into power series).

Using this notation we  have the following resolvent formula \cite{KuchmentBook, OdehKeller, RS} for the operator
$L_0(k_x)$:
\beq
(L_0(k_x) -\la)^{-1}r=\frac{1}{\sqrt{2\pi}} \sum_{s \in \N} \int_{-\pi}^\pi(\lambda_s(k)-\la)^{-1}P_s(k,r)dk
\label{eq:res1}
\enq
for $\la$ outside the spectrum of $L_0(k_x)$ and $r\in L^2((0,1)^2)$.
\section{Generation of spectrum in the gap}
\label{genofspec}
In this section we use the representation of the
resolvent by Bloch functions and the variational principle to show that
if the perturbation $\eps_1$ is of fixed sign it will lead to the generation of extra spectrum in the
gap $(\mu_0,\mu_1)$ of the spectrum of the operator $L_0(k_x)$.
For definiteness, in this section we assume that $\eps_1$ is a
non-negative function which is positive on a set of positive measure. Then
there exist $\alpha>0$ and a set $D$ of positive measure such that $\inf_D \eps_1|_{\Omega}=\alpha$.

In \eqref{FP}, set $$v:=\sqrt{\frac{\eps_1}{\eps_0}}u.$$
Then $v$ is supported in $[0,1]^2$, as $\eps_1|_\Omega$ is, and $v$ satisfies
\beq\label{eq:v2}
v=\lambda\sqrt{\frac{\eps_1}{\eps_0}}\left(L_0(k_x)-\lambda\right)^{-1}\sqrt{\frac{\eps_1}{\eps_0}}v.
\eeq
Note that vice versa, if $v$ satisfies \eqref{eq:v2}, then
$$u:=\lambda\left(L_0(k_x)-\lambda\right)^{-1}\sqrt{\frac{\eps_1}{\eps_0}}v$$
satisfies \eqref{FP} and lies in $L^2_{\eps_0}(\Omega)$. It is therefore sufficient for our purposes to study \eqref{eq:v2}.

We now define the operator $A_\lambda$ on $L^2_{\eps_0}({(0,1)^2})$ by
$$
A_\lambda v=\left( \lambda\sqrt{\frac{\eps_1}{\eps_0}}\left(L_0(k_x)-\lambda\right)^{-1}\sqrt{\frac{\eps_1}{\eps_0}} v\right)\Bigg\vert_{(0,1)^2}
$$
and note that \eqref{eq:v2} has a non-trivial solution if and only if $1$ is an eigenvalue of the operator $A_\lambda$.

  \begin{lemma} Let $\la\in\R$. Then
$  A_{\la} : L^2_{\eps_0}({(0,1)^2}) \to L^2_{\eps_0}({(0,1)^2}) $ is symmetric and compact.
  \end{lemma}

\begin{proof} Let $u,v\in L^2_{\eps_0}({(0,1)^2})$. Then
  \begin{eqnarray*}
  \llangle \eps_0 A_\la u,v\rrangle_{L^2({(0,1)^2})} &=& \llangle \eps_0 \la  \left ( -\frac1{\eps_0}\Delta -\la\right )^{-1} \sqrt{\frac{\eps_1}{\eps_0}}u, \sqrt{\frac{\eps_1}{\eps_0}} v\rrangle_{L^2(\Omega)}
   \\
  &=&\llangle \eps_0\sqrt{\frac{\eps_1}{\eps_0}} u,\la  \left ( -\frac1{\eps_0}\Delta -\la\right )^{-1} \sqrt{\frac{\eps_1}{\eps_0}} v \rrangle_{L^2(\Omega)} \\
  &=&\llangle \eps_0u,\la\sqrt{\frac{\eps_1}{\eps_0}} \left ( \left ( -\frac1{\eps_0}\Delta -\la\right )^{-1} \sqrt{\frac{\eps_1}{\eps_0}} v\right ) \Bigg \vert_{(0,1)^2} \rrangle_{L^2({(0,1)^2})}\\
   \ &=&\ \llangle \eps_0 u, A_{\la} v\rrangle_{L^2((0,1)^2)},
  \end{eqnarray*}
  so the operator is symmetric. Moreover, by standard estimates (see e.g. \cite{Evans}) for the elliptic operator on the strip,
  \bea
  \norm{ \left ( -\frac1{\eps_0}\Delta -\la \right )^{-1}\sqrt{\frac{\eps_1}{\eps_0}} u}_{H^1({(0,1)^2})}
  &\leq   \norm{ \left ( -\frac1{\eps_0}\Delta -\la \right )^{-1} \sqrt{\frac{\eps_1}{\eps_0}} u }_{H^1(\Omega)}\\
  &\leq C_{\la} \norm{  u }_{L^2(\Omega)}=C_{\la}\norm{ u}_{L^2((0,1)^2)}.
  \eea
  Thus $A_\la $ is the composition of a compact map with the continuous map of multiplication by the function $\sqrt{\frac{\eps_1}{\eps_0}}$ and
  is therefore compact as a map from $L^2({(0,1)^2}) \to L^2({(0,1)^2})$. Multiplication by the bounded and boundedly invertible weight $\eps_0$ does not change this.\hfill\end{proof}

We now investigate the dependence of the maximum eigenvalue  of $A_\la$ on $\la$. First define
  $$\kappa_{max}(\la):=\sup_{\norm{ u} \neq 0} \frac{ \llangle A_\la u,u\rrangle_{\eps_0}}{\llangle u,u\rrangle_{\eps_0}}.$$
  Then if $\kappa_{max}(\la)>0$, it is the maximum eigenvalue of $A_\lambda$.

\begin{lemma} 
\begin{enumerate}
    \item On the interval $(\mu_0,\mu_1)$ the map $\lambda \mapsto \kappa_{\max}(\la)$ is continuous and  increasing.
    \item If there exists $\lambda'\in(\mu_0,\mu_1)$ such that $\kappa_{\max}(\la')>0$, then  $\lambda \mapsto \kappa_{\max}(\la)$ is strictly increasing on $(\la',\mu_1)$.
\end{enumerate}
\end{lemma}

\begin{proof}
  (1)  $A_\la=\lambda\sqrt{\frac{\eps_1}{\eps_0}}\left(L_0(k_x)-\lambda\right)^{-1}\sqrt{\frac{\eps_1}{\eps_0}}$, so
  $\la\mapsto A_\la $ is norm continuous as a map from
  $(\mu_0,\mu_1)$ to ${\mathcal L}(L_{\eps_0}^2)$ and for any $\lambda \in (\mu_0,\mu_1)$ and any
  $\tilde{\eps}>0$ there exists $\delta>0$ such that  for $|\la-\tilde \la|<\delta$
  $$\Big \vert \llangle A_{\tilde \la}u,u\rrangle_{\eps_0}-\llangle A_\la u,u\rrangle_{\eps_0} \Big \vert   \leq \norm{ A_\la-A_{\tilde\la}}\norm{  u}^2_{\eps_0} \leq \tilde{\eps} \norm{u}^2_{\eps_0}.
  $$
  This implies  $\Big \vert \kappa_{max}(\tilde \la)-\kappa_{max}(\la)\big \vert \leq \tilde{\eps}$, so $\la \mapsto \kappa_{max}(\la)$ is continuous.

  Let $\mu_0<\la<\tilde \la<\mu_1$.
  Then
   $$\frac{\tilde \la}{\la_s(k)-\tilde \la}-\frac\la{\la_s(k)-\la}=\frac{ (\tilde \la-\la)\la_s(k)}{(\la_s(k)-\tilde \la)(\la_s(k)-\la)} \geq 0
   $$
   since $ (\la_s(k)-\tilde \la) (\la_s(k)-\la)>0$ and $\la_s(k)\geq 0$ for all $s$ and all $k$.   Thus $\la\mapsto \kappa_{max}(\la)$ is monotonically increasing. This proves the first statement of the lemma.

  (2)
  We now investigate strict monotonicity of  $\la \mapsto \kappa_{\max}(\la)$.  We note that $\la \mapsto \kappa_{\max}(\la)$ need not be
  differentiable, as it is possible for eigenvalue branches to cross.
  However, $\kappa_{\max}$ is piecewise analytic on the interval
  $[\lambda', \mu_1)$. Note that for $\lambda>\lambda'$, $\kappa_{\max}(\lambda) \geq \kappa_{\max}(\lambda') > 0$ by the monotonicity shown
  in the first part of the proof, so the point $0$ in the spectrum of $A_\lambda$ does not cause difficulties for
  the analytic continuation of the eigenvalues (the reader may refer to the discussion in Kato's book,
  \cite[Theorem VII.3.9 and Remark VII.3.11]{Kato}).

  Consider now some subinterval of $[\lambda', \mu_1)$ where $\kappa_{\max}$ is analytic
  and let $u_\la$ be the corresponding normalised eigenfunction of $A_\la$, depending analytically on $\lambda$.
  Now, $$\kappa_{\max}(\la)=\llangle  A_\la u_\la,u_\la\rrangle_{\eps_0},$$ so  using the symmetry of $A_\la$ we have

\bea
\frac{\partial \kappa_{\max}(\la)}{\partial \la} &=&
  \llangle \frac{ \partial A_\la}{\partial \la} u_\la,u_\la\rrangle_{\eps_0} + \llangle A_\la \frac{ \partial u_\la}{\partial \la},u_\la\rrangle_{\eps_0} +\llangle A_\la u_\la,\frac{\partial u_\la}{\partial \la}\rrangle_{\eps_0} \\
  &=&
  \llangle \frac{ \partial A_\la}{\partial \la} u_\la,u_\la\rrangle_{\eps_0} + \kappa_{\max}(\la)\underbrace{\left ( \llangle   \frac{ \partial u_\la}{\partial \la},u_\la\rrangle_{\eps_0}+
 \llangle u_\la , \frac{ \partial u_\la}{\partial \la} \rrangle_{\eps_0} \right )}_{= \frac{\partial}{\partial \la }\llangle u_\la,u_\la \rrangle_{\eps_0}=0}\\
 &=&\llangle \frac{ \partial A_\la}{\partial \la} u_\la,u_\la\rrangle_{\eps_0} \ =\ \frac{1}{2\pi}\sum_s \int_{-\pi}^{\pi}\frac{\la_s(k)}{(\la_s(k)-\la)^2} \Bigg \vert \llangle \sqrt{\frac{\eps_1}{\eps_0}}u_\la,\psi_s\rrangle_{\eps_0}\Bigg \vert^2  dk,
 \eea
  where in the last step we have used the representation of the resolvent via Bloch functions \eqref{eq:res1}.
  Now, $\la_s(k)>0$  for $k\neq 0$ (else the operator $L_0(k_x,k)$ would have a non-positive eigenvalue).
  Assume $\frac{\partial \kappa_{\max}(\la)}{\partial \la}=0 $. Then $\llangle \sqrt{\frac{\eps_1}{\eps_0}} u_\la,\psi_s\rrangle =0$ for a.e. $k$ and all $  s \in \N$.
  Thus $\sqrt{\frac{\eps_1}{\eps_0}} u_\la=0$ which implies $A_\la u_\la=0$. However, this cannot be the case, as $\kappa_{\max}(\la)>0$. Therefore, $\kappa_{\max}(\lambda)$ is strictly increasing on $(\la',\mu_1)$.\end{proof}

We are now ready to prove our main result which shows that
small perturbations of $\eps_0$ create eigenvalues in the spectral gap.
\begin{theorem}\label{defectmodes}
Assume that \beq \label{pertnorm} \norm{\eps_1}_\infty<\frac{(\mu_1-\mu_0)\inf \eps_0}{\mu_0}.\eeq
Then there exists an eigenvalue of the operator $L(k_x)$ in the spectral gap $(\mu_0,\mu_1)$ of $L_0(k_x)$.
\end{theorem}

\begin{proof}
We shall use the Rayleigh quotient to obtain information on the eigenvalues of $A_\la$. Let $u\in L^2({(0,1)^2})$, then using \eqref{eq:res1} we have
  \begin{eqnarray}
  &&\llangle \eps_0A_\la u,u \rrangle_{L^2({(0,1)^2})}\nonumber\\
   &=&\la \llangle \eps_0  \sqrt{\frac{\eps_1}{\eps_0}}\left [ \left ( -\frac1\eps_0\Delta -\la \right )^{-1}  \sqrt{\frac{\eps_1}{\eps_0}} u \right ]_{(0,1)^2}, u \rrangle \nonumber \\
  & =& \la \llangle \eps_0   \left ( -\frac1\eps_0\Delta -\la \right )^{-1}\sqrt{\frac{\eps_1}{\eps_0}} u,  \sqrt{\frac{\eps_1}{\eps_0}} u \rrangle_{L^2(\Omega)}
  \nonumber\\
  &= &\frac{\la}{2\pi} \int_{-\pi}^\pi \sum_{s \in \N} (\la_s(k)-\la)^{-1} \left| \llangle  \sqrt{\frac{\eps_1}{\eps_0}}u,\psi_{s}(\cdot,k)\rrangle_{L^2_{\eps_0}({(0,1)^2})}  \right|^2 dk. \label{form}
  \end{eqnarray}

To first get an upper estimate on the Rayleigh quotient, let $s'$ be such that $\mu_1$ is the lowest point of the $s'$-band and $\mu_0$ is the highest point of the $(s'-1)$-band. We note that such an $s'$ must exist for there to be a gap. Let $\lambda\in(\mu_0,\mu_1)$.
 Then since $\la_s(k)-\la\leq 0$ for $s<s'$ we have
 \ben \label{Rayleighup}
  \llangle \eps_0A_\la u,u \rrangle_{L^2({(0,1)^2})} &\leq & \frac{\la}{2\pi}\int_{-\pi}^\pi \sum_{s \geq s'} (\la_s(k)-\la)^{-1} \left| \llangle  \sqrt{\frac{\eps_1}{\eps_0}}u,\psi_{s}(\cdot,k)\rrangle_{\eps_0}  \right|^2 dk\nonumber
  \\ \nonumber
  &\leq& \frac{\la}{2\pi(\mu_1-\la)} \int_{-\pi}^\pi \sum_{s \geq s'}  \left| \llangle  \sqrt{\frac{\eps_1}{\eps_0}}u,\psi_{s}(\cdot,k)\rrangle_{\eps_0}  \right|^2 dk\\ \nonumber
  &\leq& \frac{\la}{2\pi(\mu_1-\la)} \int_{-\pi}^\pi \sum_{s\in\N}  \left| \llangle  \sqrt{\frac{\eps_1}{\eps_0}}u,\psi_{s}(\cdot,k)\rrangle_{\eps_0}  \right|^2 dk\\
  &=& \frac{\la}{(\mu_1-\la)} \norm{  \sqrt{\frac{\eps_1}{\eps_0}}u}_{\eps_0}^2 \ \leq\ \frac{\la\norm{\eps_1}_\infty}{(\mu_1-\la)\inf\eps_0} \norm{u}_{\eps_0}^2 .
 \een

 Therefore, if the perturbation $\eps_1$ is sufficiently small such that \eqref{pertnorm} holds, we can find $\lambda'\in(\mu_0,\mu_1)$ such that
 \beq\label{upperkappa} \kappa_{max}(\la')=\sup_{\norm{u}\neq 0} \frac{ \llangle A_{\la'} u ,u \rrangle_{\eps_0} }{\norm{u }_{\eps_0}^2} < 1.\eeq

 We next want to find a lower bound on the Rayleigh quotient.
Let $k_0$ be such that $\lambda_{s'}(k_0)=\mu_1$.

We remind the reader that $\eps_1\geq\alpha>0$ on a set $D$ of positive measure.
 It follows from unique continuation (see e.g. \cite{Protter})
 that for any Bloch function $\psi_s(\cdot,k)$ we have $$\int_D |\psi_s(\cdot,k)|^2 \eps_0>0 .$$
Therefore, and by continuity of the Bloch functions there exist $\delta>0$ and $a>0$ such that
  $$
  |  \llangle \psi_{s'}(\cdot, k_0),\psi_{s'}(\cdot,k)\rrangle_{L^2_{\eps_0}(D)} |^2\geq a
  $$
  for $k \in (k_0-\delta, k_0+\delta)$.
Now choose $u=\sqrt{\frac{\eps_0}{\eps_1}}\psi_{s'}(\cdot,k_0)$ on $D$ and extend $u$ by zero to $(0,1)^2$.
  Then from \eqref{form} we get that for $\lambda\in(\mu_0,\mu_1)$
  \bea
 \frac{ \llangle \eps_0 A_\la u,u \rrangle}{\norm{u}_{\eps_0}^2} &=&
 \frac{\la}{2\pi\norm{u}_{\eps_0}^2} \int_{-\pi}^\pi \sum_{s \in \N} (\la_s(k)-\la)^{-1} \left| \llangle \psi_{s'}(\cdot,k_0),\psi_{s}(\cdot,k)\rrangle_{L^2_{\eps_0}(D)}  \right|^2 dk.  \\
 &\geq&    \frac{\la}{2\pi\norm{u}_{\eps_0}^2} \int_{-\pi}^\pi \sum_{s \leq s'} (\la_s(k)-\la)^{-1} \left| \llangle \psi_{s'}(\cdot,k_0),\psi_{s}(\cdot,k)\rrangle_{L^2_{\eps_0}(D)}  \right|^2 dk.  \\
   &\geq& \frac{a\la}{2 \pi\norm{u}_{\eps_0}^2}
   \int_{k_0-\delta}^{k_0+\delta}\frac{dk}{ \la_{s'}(k)-\la} -\\
   &&\frac{1}{2\pi\norm{u}_{\eps_0}^2} \int_{-\pi}^\pi \sum_{s < s'} \frac{\la}{\la-\la_s(k)} \left| \llangle \psi_{s'}(\cdot,k_0)  ,\psi_{s}(\cdot,k)\rrangle_{L^2_{\eps_0}(D)}  \right|^2 dk\\
  &\geq& \frac{a\la}{2 \pi\norm{u}_{\eps_0}^2} \int_{k_0-\delta}^{k_0+\delta} \frac {dk}{ \la_{s'}(k)-\la}-C_\la,
  \eea
where $C_\la$ is bounded as long as $\la$ stays away from $\mu_0$.
Using the Taylor expansion of the analytic function $\la_{s'}$, which has a minimum at $k_0$,
  we see that $$\la_{s'}(k)\leq \mu_1+\alpha_n(k-k_0)^{2n}$$ for  $k \in (k_0-\delta, k_0+\delta)$ and
  for some $n \in \N$ with
  $$\alpha_n=\frac1{(2n)!} \max_{\xi \in [k_0-\delta,k_0+\delta]}\la_{s'}^{(2n)}(\xi) >0.$$
  Thus for $\delta<1$
  \begin{eqnarray*}
  \int_{k_0-\delta}^{k_0+\delta}  \frac{dk}{\lambda_{s'}(k)-\la} &\geq&   \int_{k_0-\delta}^{k_0+\delta}  \frac{dk}{\mu_1-\la+\alpha_n(k-k_0)^{2n}}\
  \geq\    \int_{k_0-\delta}^{k_0+\delta}  \frac{dk}{\mu_1-\la+\alpha_n(k-k_0)^2}
 \\
 &=& 
 \frac1{\mu_1-\la} \arctan \left ( \sqrt{\frac{\alpha_n}{\mu_1-\la}} X\right ) \sqrt{ \frac{ \mu_1-\la}{\alpha_n}}\Big\vert_{X=-\delta}^{X=\delta}\\
&=& \frac{2}{\sqrt{\alpha_n( \mu_1-\la)}} \arctan  \left ( \sqrt{\frac{\alpha_n}{\mu_1-\la} }\delta\right )  \to  \infty  \hbox{ as } \lambda\nearrow\mu_1 
  \end{eqnarray*}
  and hence $$\kappa_{max}(\la)\geq  \frac{ \llangle \eps_0 A_\la u,u \rrangle}{\norm{u}_{\eps_0}^2}  \to  \infty  \hbox{ as } \lambda\nearrow\mu_1.$$

Since by \eqref{upperkappa} we have some $\la'$ with  $\kappa_{max}(\la')<1$ and $\kappa_{max}(\cdot)$ is continuous, we can find $\la\in(\la',\mu_1)$ such that $\kappa_{max}(\lambda)=1$. From \eqref{FP} and \eqref{eq:v2} we see
that this gives an eigenvalue of the perturbed strip operator $L(k_x)$ between $\la'$ and $\mu_1$.\end{proof}

\begin{remark}
\begin{enumerate}
    \item Our result only holds for small perturbations. One would of course generally expect that larger perturbations lead to the creation of more spectrum in the gap as lower eigenvalues of $A_\la$ cross $1$, but our method of proof does not cover those cases. However, existence of guided modes for sufficiently large perturbations is shown for example in \cite{FK97,MiaoMa}.
    \item We consider a special case:
Assume that $\eps_1(x)=\alpha \tilde{\eps}(x)$ for some function $\tilde{\eps}(x)\geq0$ and $\alpha\in\R^+$. This enables us to study the spectrum in terms of the scalar parameter $\alpha$. We can interpret higher values of $\alpha$ as switching on or increasing the perturbation. Let $\At_\la=A_\la/\alpha$ and $\nu_{\max}(\la)$ be the maximum eigenvalue of $\At_\la$. Then we have $\la\in\sigma(L(k_x))$ iff $\alpha\nu_{\max}(\lambda)=1$.
This gives the following results:
\begin{enumerate}
    \item   Monotonicity of $\lambda\mapsto\nu_{\max}(\la)$ and the fact that the maximal eigenvalue $\nu_{\max}(\la)$ of $\At_\la$ tends to $+\infty$ as $\la\to \mu_1$
 implies that for any given small $\alpha$, we can find $\la$ such that $\nu_{\max}(\la)=\frac1{\alpha}$, so an additional spectral point is produced  as soon as the perturbation is switched on.
 \item   Strict monotonicity of $\lambda\mapsto\nu_{\max}(\la)$ implies that
  there  exist $\la'>\mu_0$ and $\beta >0$ such that $\nu_{\max}:(\la',\mu_1)\to (\beta,\infty)$ is invertible.
  Given $\alpha < \frac1\beta$,  the function $g(\alpha):=\nu_{\max}^{-1}(\frac1{\alpha})$ always gives an eigenvalue and $g(\alpha)\to \mu_1$ when $\alpha\to 0$.
 \item    Let $\la \in(\mu_0,\mu_1)$. The estimate \eqref{Rayleighup} shows that
 the maximum eigenvalue $\nu_{\max}(\la)$ of $\At_\la$ satisfies
   $$  \nu_{\max}(\la)\leq \frac{\la\norm{\tilde{\eps}}_\infty}{(\mu_1-\la)\inf\eps_0}.$$
      Hence, if $ \alpha < \frac{(\mu_1-\la)\inf\eps_0}{\la\norm{\tilde{\eps}}_\infty}$ then
  $\alpha \nu_{\max}(\la) <1$,
  so $\la$ is not an eigenvalue.
  This implies that the perturbation needs to have a certain size before the spectrum can appear at any  given   point in the gap away from $\mu_1$.
\end{enumerate}
\item  In the case when $\eps_1$ is a non-positive function analogous results can be shown by considering the operator
 $$ A'_\lambda v= -\lambda\sqrt{\frac{-\eps_1}{\eps_0}}\left(L_0(k_x)-\lambda\right)^{-1}\sqrt{\frac{-\eps_1}{\eps_0}} v.$$
 In this case, the extra eigenvalues created for small perturbations will appear at the bottom end of the spectral gap near $\mu_0$.
\end{enumerate}
\end{remark}

\section{Non-accumulation of the eigenvalues}
\label{nonaccum}

In this section we wish to show that new eigenvalues generated by the perturbation cannot accumulate at the band edges.
Unlike in section \ref{genofspec}, we no longer make any assumptions on the sign of the perturbation.

In \eqref{FP}, set $$v=\frac{\eps_1}{\eps_0}u.$$
Note that $v$ is supported in $[0,1]^2$, as $\eps_1\vert_\Omega$ is, and $v$ satisfies
\beq\label{eq:v}
v=\lambda\frac{\eps_1}{\eps_0}\left(L_0(k_x)-\lambda\right)^{-1}v.
\eeq
Therefore, we shall study the spectrum of the operator $\lambda\frac{\eps_1}{\eps_0}\left(L_0(k_x)-\lambda\right)^{-1}$ acting on functions supported in $[0,1]^2$.

\subsection{Analytic continuation of the resolvent}

Recall the  resolvent formula \eqref{eq:res1}:
\beq
(L_0(k_x) -\la)^{-1}r=\frac{1}{\sqrt{2\pi}} \sum_{s \in \N} \int_{-\pi}^\pi(\lambda_s(k)-\la)^{-1}P_s(k,r)dk
\label{eq:res2}
\eeq
 for $ \la$ outside the spectrum of $L_0(k_x)$ and $r\in L^2(\Omega)$ with compact support.
Restricted to functions with compact support, the resolvent $(L_0(k_x) -\la)^{-1}$ is compact and we would like to use meromorphic Fredholm theory (see Theorem \ref{Fredholm}) to
analyze the spectrum. However, the resolvent is not well-defined in a neighborhood of the band edge. To overcome this difficulty, we transform the critical integrals in \eqref{eq:res2} by integration over a suitable contour in the complex plane (see Figure \ref{fig2}) instead of the real interval $[-\pi,\pi]$. In this way, we obtain
an analytic operator family, see \eqref{resolvent}, that coincides with the resolvent on a sector of the complex plane.
We will do this construction for the lower edge of a band near $\mu_1$. A similar construction is
possible near the top end of a band. 


Note that by  Proposition \ref{finite} there are only finitely many pairs $(s_p,k_p)$, $p=1,...,N$ such that for $\lambda=\mu_1$ the integrand
$(\lambda_{s_p}(k)-\la)^{-1}P_{s_p}(k,r)$ in  \eqref{eq:res2} is singular at $k_p$, since $\lambda_{s_p}(k_p)=\mu_1$.
As none of the  analytic  functions $\la_s$ can have a zero of infinite multiplicity by the Thomas argument (see \cite{KuchmentBook,Thomas}), we have for each $p$ that $\lambda_{s_p}^{(j)}(k_p)= 0$ for $j=1,...,m_p-1$ and $\lambda_{s_p}^{(m_p)}(k_p)\neq 0$ for some even $m_p\geq 2$. Since $\mu_1$ lies at the bottom end of the band,
we can write $\lambda_{s_p}(k)=\mu_1+(k-k_p)^{m_p}g_p(k)$ with $g_p$ an analytic function of $k$, $g_p(k_p)> 0$ and $g_p(\R)\subseteq\R$.

We set $h_{s_p,k_p}(k):=(k-k_p)\sqrt[m_p]{ g_p(k)}$. Here, we choose the branch cut of the root away from the positive real axis, so that it does not intersect with the set $\{g_p(k)\}$ for $k$ in a small neighborhood of $k_p$. Since $h'_{s_p,k_p}(k_p)=\sqrt[m_p]{g_p(k_p)}\neq 0$, we can locally invert $h_{s_p,k_p}$ to obtain $k=h^{-1}_{s_p,k_p}(\nu)$, and $k$
will depend analytically on $\nu$ for $\nu$ in a 
ball $\cN$ centered at $0$. Let $\cN^+=\cN\cap\{\Im(\nu)>0\}$ and $\cN^-=\cN\cap\{\Im(\nu)<0\}$.
 Clearly, $h_{s_p,k_p}:\R\to\R$, so $h^{-1}_{s_p,k_p}:\cN^\pm\to\C\setminus\R$. As $\cN^\pm$ are connected, $h^{-1}_{s_p,k_p}(\cN^\pm)$ are connected and are therefore completely contained in either the lower or upper half-plane. Take $\nu=\pm is\in\cN^\pm$. Then, as $$\left(h^{-1}_{s_p,k_p}\right)'(0)=1/\sqrt[m_p]{g_p(k_p)}> 0,$$
 by Taylor series expansion we get
$$h^{-1}_{s_p,k_p}(\pm is)=k_p+\frac{\pm is}{\sqrt[m_p]{g_p(k_p)}}+r_\pm(is)$$ with $r_\pm(is)/s\to 0 $ as $s\to 0$. Hence, $\Im(h^{-1}_{s_p,k_p}(is))>0$ and $\Im(h^{-1}_{s_p,k_p}(-is))<0$, so $h^{-1}_{s_p,k_p}$ maps $\cN^+$ to the upper and $\cN^-$ to the lower half-plane.

\begin{prop}\label{solutions}
All solutions of $(k-k_p)^{m_p} g_p(k)= \nu^{m_p}$ in a neighborhood of $k_p$ are given by $$k= h^{-1}_{s_p,k_p}(e^{2\pi ip/{m_p}}\nu)$$ with  $p=0,...,m_p-1$.
\end{prop}
\begin{proof} In a punctured neighborhood of $k_p$ we have the following equivalences:
\bea
(k-k_p)^{m_p} g_p(k)= \nu^{m_p} &
\Leftrightarrow& \left(\frac{\nu}{(k-k_p) \sqrt[m_p]{g_p(k)}}\right)^{m_p}= 1\\
&\Leftrightarrow& \frac{\nu}{(k-k_p) \sqrt[m_p]{g_p(k)}}= e^{-2\pi ip/m_p},\ p\in\{0,...,m_p-1\}\\
&\Leftrightarrow& k= h^{-1}_{s_p,k_p}(e^{2\pi ip/m_p}\nu),\ p\in\{0,...,m_p-1\}.
\eea
\end{proof}

For simplicity of notation, from now on we restrict ourselves
to the case when only one band, which we call the $s_0$-band, touches $\mu_1$.
All results generalise in the obvious way in the case when more than one band touches $\mu_1$.

Let $(s_0,k_p)$, $p=1,...,N$ be all pairs described in Proposition \ref{finite} with $m_p$ being the order of the first non-vanishing derivative of
$\lambda_{s_0}$ at $k_p$. Let $m$ be the lowest common multiple of the $m_p$ and $q_p=m/m_p$.
Using the Taylor expansion of $\lambda_{s_0}(k)$ around each of these $k_p$, we find a  complex   neighborhood $N(\mu_1)$ of $\mu_1$ and balls $B_R(k_1),...,B_R(k_N)$ of some radius $R$ around each of the $k_p$ such that by Proposition \ref{solutions} the equation $\lambda_{s_0}(k)= \mu_1+\mu^{m_p}$ has precisely $m_p$ solutions in $B_R(k_p)$ whenever  $ \mu_1+\mu^{m_p} \in N(\mu_1)$.
Moreover, we can find a smaller neighborhood $\widetilde{N}(\mu_1)$ of $\mu_1$ such that for each $i=1,...,N$ and all $\mu_1+\mu^{m_p} \in \widetilde{N}(\mu_1)$ all solutions  in $B_R(k_p)$   of $\lambda_{s_0}(k)=\mu_1+\mu^{m_p}$ in fact lie in  a ball $B_{R/3}(k_p)$ of radius $R/3$  around   $k_p$.

We now choose a contour $G$ as indicated in Figure \ref{fig2}.
\begin{figure}[hbt] 
\begin{center}
\begin{picture}(0,0)%
\includegraphics{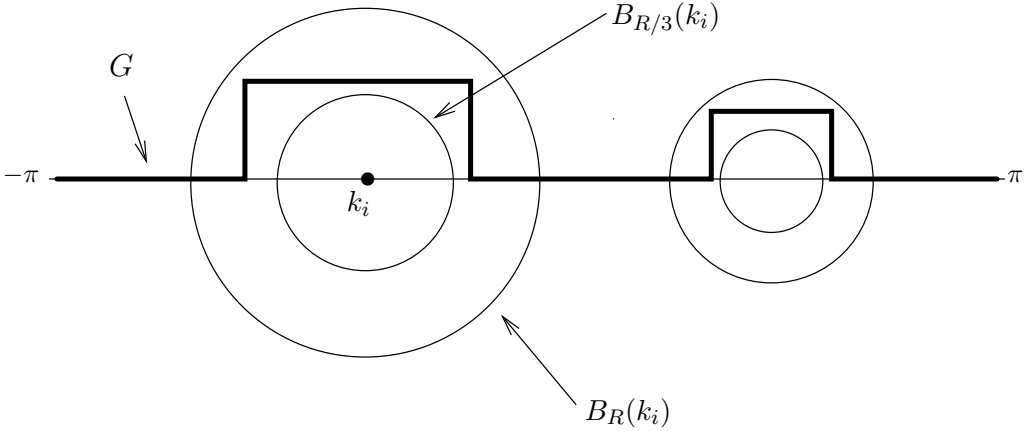}%
\put(-160, 10){\large\text{$B_R (k_{i})$}}
\put(-150, 160){\large\text{$B_{R/3}(k_{i})$}}
\put(0, 100){\text{$\pi$}}
\put(-380, 100){\text{$-\pi$}}
\put(-340, 140){\Large\text{$G$}}
\put(-242, 101){\circle*{5}}
\put(-250, 89){\large\text{$k_{i}$}}
\end{picture}%
\setlength{\unitlength}{4144sp}%
\begingroup\makeatletter\ifx\SetFigFont\undefined%
\gdef\SetFigFont#1#2#3#4#5{%
  \reset@font\fontsize{#1}{#2pt}%
  \fontfamily{#3}\fontseries{#4}\fontshape{#5}%
  \selectfont}%
\fi\endgroup%
\begin{picture}(5787,2883)(346,-2695)
\end{picture}%
\end{center}
\caption{The contour $G$.\label{fig2}}
\end{figure}

For later estimates, it will be useful to have a bound on the distance from the curve $G$ to solutions of the equation $\lambda_{s_0}(k)=z$. This is given by the following lemma.

\begin{lemma}\label{dist}
 There exists another neighborhood $\widehat{N}\subset \widetilde{N}(\mu_1)$ such that for some positive number $\delta_0$ we have \mbox{$\mathrm{dist} (G,\{ k\in\C:\lambda_{s_0}(k)=z\})\geq \delta_0$}  for all $z \in \widehat{N}$.
\end{lemma}

 \begin{proof}
 For a contradiction, we assume that for all neighborhoods $\widehat{N}$ and for all $\delta_0>0$ there exists $z \in \widehat{N}$ such that $\mathrm{dist} (G,\{ k:\lambda_{s_0}(k)=z\}) < \delta_0$. This implies that there exists a sequence
 $(z_n)$ in $ \widetilde N(\mu_1)$ such that $z_n\to\mu_1$ and
 $$\mathrm{dist} (G,\{k:\lambda_{s_0}(k)=z_n\})\to 0$$ as $ n \to\infty$.  This in turn implies the existence of a sequence
 $(\widehat{k_n})$ such that $\lambda_{s_0}(\widehat{k_n})=z_n$ and $\mathrm{dist} (G, \{\widehat{k_n}\})\to 0$ as $n \to \infty$.

 Assume for another contradiction that $\widehat{k_n} \in B_R(k_p)$ for some $i$. Then by definition of $\widetilde N(\mu_1)$ we would have $\widehat{k_n} \in B_{R/3}(k_p)$ for large $n$. However, this contradicts that the $\widehat{k_n}$ are very close to the curve $G$ for large $n$. We can deduce that $\widehat{k_n} \notin B_R(k_p)$ and since $\mathrm{dist} (G, \{\widehat{k_n}\})\to 0$ we get that $\Im\ \widehat{k_n}\to 0$ as $ n\to \infty$.

 Now, $[-\pi,\pi]\backslash \cup (B_R(k_p)\cap\R)$ is compact, so there exists a subsequence $\widehat{k_{n_j}}$ converging to some  $\hat k \in [-\pi,\pi] \backslash \cup (B_R(k_p)\cap\R)$.
 It then follows that $\lambda_{s_0}(\widehat{k_{n_j}})\to \lambda_{s_0}(\hat k)$, and $z_{n_j}\to \mu_1$.  So $\lambda_{s_0}(\hat k)=\mu_1$.  Hence, $\hat k$ is an additional real solution to $\lambda_{s_0}(k)=\mu_1$ contradicting the assumption that the $k_p$ are all solutions.\end{proof}

We are now in a position to introduce the operator family that will provide an analytic extension of the resolvent.
Let $\cO$ be a small neighborhood of $0$ such that $\mu_1+\mu^m\in \widehat{N}$ for any $\mu\in \cO$. For $\mu\in \cO\setminus\{0\}$ and $r\in L^2_{\eps_0}((0,1)^2)$ we define the following operator family
\ben
\label{resolvent}\nonumber B_\mu(r)&=&\frac{1}{\sqrt{2 \pi}}\sum_{s\not = s_0} \int_{-\pi}^\pi(\lambda_s(k)-\mu_1-\mu^m)^{-1}P_s(k,r)dk
\\
&& \ +\ \frac{1}{\sqrt{2 \pi}}\int_{G} (\lambda_{s_0}(k)-\mu_1-\mu^m)^{-1}P_{s_0}(k,r)dk\\
&& \ +\  \sqrt{2 \pi} i\sum_{i=1}^N \sum_{p\in P_i} \lambda'_{s_0}(h^{-1}_{s_0,k_i}(e^{2\pi ip/m_i}\mu^{q_i}))^{-1}P_{s_0}(h^{-1}_{s_0,k_i}(e^{2\pi ip/m_i}\mu^{q_i}),r).\nonumber
\een
where, $P_i=\{0,...,\frac{m_i}{2}-1\}$.
The next result shows that $B_\mu$ does indeed extend the resolvent.
\begin{prop}
For $\mu\in {\cO} \setminus\{0\}$  with $\arg \mu\in(0,2\pi/m)$, we have $$B_\mu(r)= (L_0(k_x) -\mu_1-\mu^m)^{-1}r$$ for $r\in L^2_{\eps_0}((0,1)^2)$
\end{prop}

\begin{proof}
Noting the resolvent formula \eqref{eq:res1}, we see that $B_\mu$ is obtained from the resolvent by replacing integration over $(-\pi,\pi)$ with integration over $G$ when $s=s_0$  and adding the final term with the point evaluation.
Hence, this amounts to a simple application of the residue theorem:  We need to find the solutions to
$\lambda_{s_0}(k)=\mu_1+\mu^m$, i.e.~to  $(k-k_i)^{m_i} g(k)= \mu^{m}$, lying between $G$ and the real axis.
Since $h^{-1}_{s_0,k_i}$ maps points in the upper half-plane to points in the upper half-plane and points in the lower half-plane to  points in the lower half-plane,  by Proposition \ref{solutions} and the discussion following it, these are given by $k= h^{-1}_{s_0,k_i}(e^{2\pi ip/m_i}\mu^{m/m_i})$ for those $p\in\{0,...,m_i-1\}$ such that $e^{2\pi ip/m_i}\mu^{m/m_i}$ lies in the upper half plane. Now for $\arg \mu\in(0,2\pi/m)$,
$$ \arg \left(e^{2\pi ip/m_i}\mu^{m/m_i}\right) \in  \left(2\pi \frac{p}{m_i},2\pi \frac{p}{m_i}+\frac{2\pi}{m_i}\right).  $$
This is in $(0,\pi)$ if and only if $p\in P_i$.
Note that in calculating the residue, we have used
that $\mathrm{Res}_z f/g = f(z)/g'(z)$ whenever $g(z)=0,\;g'(z)\neq 0$.\end{proof}

\subsection{Properties of the operator family $B_\mu$}

We next look at properties of the operator family $B_\mu$ from \eqref{resolvent} with $\mu$ in the small neighborhood $\cO$ of $0$. Recall that our aim is to be able to apply Fredholm theory to this operator family.

\begin{prop}
$\mu\mapsto B_\mu$ is analytic as a map from $\cO\setminus \{0\}$ to $\cL(L^2_{\eps_0}((0,1)^2))$.
\end{prop}

\begin{proof} We discuss the three terms on the right hand side of \eqref{resolvent} separately.

(1) The only problem in the residue term arises when $\lambda'_{s_0}(h^{-1}_{s_0,k_i}(e^{2\pi ip/m_i}\mu^{q_i}))=0$. Now,
$\mu=0$ is locally the only solution of $ h^{-1}_{s_0,k_i}(\mu)=k_i$ and, since
$$\lambda'_{s_0}(k)=(k-k_i)^{m_i-1} \underbrace{[m_ig(k)+(k-k_i)g'(k)]}_{\neq 0 \hbox{ for } k=k_i} ,$$
$k_i$ is locally the only zero of $k\mapsto \lambda'_{s_0}(k)$.  Therefore, the residue term is analytic for $\mu\in \cO\setminus\{0\}$.

(2) We next consider the contour integral
 \beq\label{contour}\frac{1}{\sqrt{2 \pi}} \int_{G} (\lambda_{s_0}(k)-\mu_1-\mu^m)^{-1} P_{s_0}(k,r) dk.\eeq
Lemma \ref{dist} implies that along the contour $G$ we have that $|(\lambda_s(k)-\lambda_{s_0}-\mu^m)^{-1}|$ is bounded.  In particular, this implies that the integrand in \eqref{contour} is analytic in $\mu$.

(3) Finally, we consider the infinite sum component in the resolvent. We note that by an argument as in the proof of Proposition \ref{finite}, there exists $\eta >0$ such that for all  band functions  $\lambda_s$   which do not touch
$\mu_1$ and for all $k$,  we have $|\lambda_s(k)-\mu_1|\geq \eta$. So there
exists a neighborhood $\cO$ of the origin such that for all $\mu \in \cO$ and for all non-touching bands $s$ we have $|\lambda_s(k)-\mu_1-\mu^m|\geq \frac{\eta}{2}$ for all $k$.
Hence, $|(\lambda_s(k)-\mu_1-\mu^m)^{-1}|$ is uniformly bounded in $s$ and $\mu$. We need to test with a function $\phi$ and show that for any $\phi$ the resulting function is analytic in $\mu$. Fubini's Theorem allows us to interchange the integration in $k$ and the $L^2$-scalar products. Then we have for $\mu \in \cO$ and $M\geq N$ sufficiently large, 
 \begin{eqnarray*}
 && \sum_{s=N}^M\int_{-\pi}^\pi\frac{1}{\lambda_s(k)-\mu_1-\mu^m} \llangle Ur(\cdot,k),\psi_s(\cdot,k)\rrangle_{\eps_0}\llangle\psi_s(\cdot,k),\phi\rrangle_{\eps_0} dk \\
 && \leq \sum_{s=N}^M\int_{-\pi}^\pi\frac{2}{\eta} | \llangle Ur(\cdot,k),\psi_s(\cdot,k)\rrangle_{\eps_0}|\ |\llangle\psi_s(\cdot,k),\phi\rrangle_{\eps_0}|\ dk \\
 &&  \leq
 \frac{2}{\eta}\sqrt{ \sum_{s=N}^M\int_{-\pi}^\pi | \llangle Ur(\cdot,k),\psi_s(\cdot,k)\rrangle|_{\eps_0}^2 dk}\sqrt{  \sum_{s=N}^M \int_{-\pi}^\pi | \llangle\psi_s(\cdot,k),\phi\rrangle|_{\eps_0}^2 dk}.
 \end{eqnarray*}
By completeness of the Bloch functions (see \eqref{complete}), this tends to $0$ as $N,M\to\infty$.
Hence we get uniform convergence of the series in $\mu$ and the limit function is analytic.\end{proof}

\begin{prop}
$B_\mu$ is compact for $\mu\in \cO\setminus \{0\}$.
\end{prop}

\begin{proof}
Again, we consider the three terms on the r.h.s.~of \eqref{resolvent} separately.

(1) The term from the residue is clearly compact, as it is a finite rank operator.

(2) We next consider the contour integral \eqref{contour}. Lemma
\ref{dist} implies that along the contour $G$ we have that
$|(\lambda_{s_0}(k)-\mu_1-\mu^m)^{-1}|$ is bounded and the same is true of $\norm{ \psi_{s_0}(\cdot,k)}$. Using \eqref{eq:Pskr1}
 and taking the norm into the integral,
we can now estimate
\begin{eqnarray*}
&\norm{ \int_{G} (\lambda_{s_0}(k)-\mu_1-\mu^m)^{-1}P_{s_0}(k,r)(\cdot)dk}_{H^1((0,1)^2)}\\
&\leq C \int_{G}\frac{ \norm{ \psi_{s_0}(\cdot,k)}_{H^1((0,1)^2)}}{|\lambda_{s_0}(k)-\mu_1-\mu^m|} | \llangle r,\psi_{s_0}(\cdot,k)\rrangle_{\eps_0}| dk \\
&\leq C \int_{G}\frac{ \sqrt{\lambda_{s_0}(k)}+1}{|\lambda_{s_0}(k)-\mu_1-\mu^m|}  | \llangle r,\psi_{s_0}(\cdot,k)\rrangle_{\eps_0}| dk
&\leq C\norm{r}_{\eps_0}.
\end{eqnarray*}
where the second inequality follows using \eqref{psi_s_estimate}. Hence, as an operator in $L^2_{\eps_0}((0,1)^2)$, this part of the resolvent is compact.

(3)
Finally, we consider the infinite sum component in the resolvent. We note that, as before, there exists $\eta >0$ such that for all $\mu \in \cO$ and for all non-touching bands $s$ we have $|\lambda_s(k)-\mu_1-\mu^m|\geq \frac{\eta}{2}$ for all $k$.

Hence, $|(\lambda_s(k)-\mu_1-\mu^m)^{-1}|$ is uniformly bounded in $s$ and $\mu$ and
\begin{eqnarray*}
 &&\norm{ \sum_{s=N}^M\int_{-\pi}^\pi (\lambda_s(k)-\mu_1-\mu^m)^{-1} P_s(k,r)(:) dk}^2_{H^1}\\
 &=&\norm{ \int_{-\pi}^\pi \sum_{s=N}^M ( \lambda_s(k)-\mu_1-\mu^m)^{-1} \llangle Ur(\cdot,k),\psi_s(\cdot,k)\rrangle_{\eps_0}\psi_s(:,k)dk}^2_{H^1}\\
  &\leq& C\int_{-\pi}^\pi\norm{\sum_{s=N}^M(\lambda_s(k)-\mu_1-\mu^m)^{-1}\llangle Ur(\cdot,k),\psi_s(\cdot,k)\rrangle_{\eps_0}\psi_s(:,k)}_{H^1}^2 dk \\
    &\leq&C\int_{-\pi}^\pi\sum_{s=N}^M\frac{\norm{\psi_s(:,k)}_{H^1}^2}{(\lambda_s(k)-\mu_1-\mu^m)^2}|\llangle Ur(\cdot,k),\psi_s(\cdot,k)\rrangle_{\eps_0}|^2 dk \\
  &\leq&C\int_{-\pi}^\pi\sum_{s=N}^M\frac{\left(\sqrt{\lambda_s(k)}+1\right)^2}{(\lambda_s(k)-\mu_1-\mu^m)^2}|\llangle Ur(\cdot,k),\psi_s(\cdot,k)\rrangle_{\eps_0}|^2 dk \\
 &\leq& C \sum_{s=N}^M\int_{-\pi}^\pi|\llangle Ur(\cdot,k),\psi_s(\cdot,k)\rrangle_{\eps_0}|^2 dk\to 0
 \end{eqnarray*}
as $M,N\to\infty$. In the third line we have used the $H^1$-orthogonality of the eigenfunctions
$\{\psi_s(\cdot, k) \}_{s\in \N}$ for fixed $k$. Thus, we have a
Cauchy sequence in $H^1((0,1)^2)$. Moreover, the $H^1$-norm of the limit is bounded by $C\norm{r}_{\eps_0}$ which gives compactness as an operator in $L^2_{\eps_0}((0,1)^2)$.\end{proof}

\begin{prop}
$B_\mu$ has a pole of finite rank for $\mu=0$.
\end{prop}

\begin{proof}
We note that the only pole comes from the residue term at $\mu=0$.
The operator has a pole of finite order as $\lambda'_{s_0}$ only has simple
zeroes of order $m_i-1$ at the $k_i$ and  $h^{-1}_{s_0,k_i}$ is analytic in a
neighborhood of $0$. Moreover, the factors $P_{s_0}(h^{-1}_{s_0,k_i}(e^{2\pi ip/m_i}\mu^{q_i}),r) $ are of rank 1.
\end{proof}

\begin{prop} \label{triv}
Let $\mu\in \cO\setminus\{0\}$ such that $\arg \mu\in(0,\pi/m)$. Let
\beq\label{Bmutilde}\widetilde{B}_\mu= (\mu_1+\mu^m)\frac{\eps_1}{\eps_0}B_\mu.\eeq
 Then $(I-\widetilde{B}_\mu)v=0$ only has the trivial solution in $L^2_{\eps_0}((0,1)^2)$.
\end{prop}

\begin{proof}
Suppose that there exists $v\neq 0$ such that $(I-\widetilde{B}_\mu)v=0$. Set $w:=\left(L_0(k_x) -\lambda\right)^{-1}v$ where $\lambda=\mu_1+\mu^m$. Then the equation for $v$ implies that $v=\lambda\frac{\eps_1}{\eps_0}w$, so $\left(-\frac{1}{\eps_0}\Delta -\lambda\right)w=v=\lambda\frac{\eps_1}{\eps_0}w$, or $-\Delta w = \lambda (\eps_0 + \eps_1) w$. But this implies that $\lambda=\mu_1+\mu^m$ is a non-real eigenvalue of the operator $L(k_x)$, yielding a contradiction.
 \end{proof}

\subsection{Main result}
To prove our result on the non-accumulation of eigenvalues, we make use of the following theorem on meromorphic Fredholm theory \cite[Theorem XIII.13]{RS}:

\begin{theorem}\label{Fredholm}
Let $D$ be a domain in $\C$, $S$ a discrete subset of $D$, and $H$ a Hilbert space. Assume we have a family of operators $\{A_z: z\in D\}$ such that
\begin{enumerate}
    \item\label{i} $z\mapsto A_z$ is analytic as a map from $D\setminus S$ to $\cL(H)$,
    \item\label{ii} $A_z$ is compact for $z\in D\setminus S$,
    \item\label{iii} $A_z$ has poles in $S$ of finite rank,
    \item\label{iv} there exists $z\in D\setminus S$ such that $(I-A_z)u=0$ has only the trivial solution.
\end{enumerate}
Then there exists at most a discrete set $\tilde{S}\subset D$ such that $(I-A_z)^{-1}\in\cL(H)$ for all  $z\in D\setminus(S\cup\tilde{S})$.
\end{theorem}

We can therefore prove our main result of this section:
\begin{theorem} \label{mr}
The spectrum of the problem on the strip, i.e.~the spectrum of the operator $L(k_x)$ can not accumulate at the ends of bands of the spectrum of the operator $L_0(k_x)$.
\end{theorem}

\begin{proof}
By \cite[Lemma 10]{AS}, the spectrum of the operator $L(k_x)$ outside the bands can only consist of eigenvalues. By \eqref{eq:v}, it is clear that $\lambda=\mu_1+\mu^m  \in\sigma_p(L(k_x))$ if and only if $I-\widetilde{B}_\mu$ is not invertible.

We apply Theorem \ref{Fredholm} to the operator family $\widetilde{B}_\mu$ from \eqref{Bmutilde} with $\mu$ in a small neighborhood $\cO$ of $0$.
The properties (1)-(3) follow from the propositions proved for $B_\mu$ in the previous sub-section which obviously carry over to $\widetilde{B}_\mu$, while (4) is shown in Proposition \ref{triv}.
\end{proof}
\subsection*{Acknowledgements} The authors would like to thank T. Dohnal, who showed them
a numerical example illustrating the results of section 4.

\end{document}